\documentclass{ecai}
\usepackage{times}
\usepackage{graphicx}
\usepackage{latexsym}

\usepackage{framed}
\usepackage{enumerate}
\usepackage{subfigure}
\usepackage{amsmath}
\usepackage{booktabs}

\usepackage{amsfonts}
\usepackage{amssymb}
\usepackage{mathrsfs}
\usepackage{amsthm}
\usepackage{color}

\newtheorem{definition}{Definition}

\newtheorem{thm}{Theorem}

\begin{document}

\title{Incentivize Diffusion with Fair Rewards}

\author{Wen Zhang \and Dengji Zhao \and Yao Zhang\institute{ShanghaiTech University,
		China, email: \{zhangwen, zhaodj, zhangyao1\} @shanghaitech.edu.cn}}

\maketitle
\bibliographystyle{ecai}

\begin{abstract}
  This paper studies a sale promotion mechanism design problem on a social network, where a node (a seller) sells one item to the other nodes on the network to maximize her revenue. However, the seller does not know other nodes except for her neighbours and her neighbours have no incentive to promote the sale. Hence, the goal is to design an auction mechanism such that the seller's neighbours are incentivized to invite their neighbours to join the auction, while the seller's revenue is guaranteed to increase. This is not achievable with traditional mechanisms. One solution has been proposed recently by carefully designing a reward scheme for the nodes who have invited others. However, the solution only gives rewards to some cut-points of the network, but cut-points rarely exist in a well-connected network, which actually disincentivizes nodes' participation. Therefore, we propose another novel mechanism to reward more related participants with fairer rewards, and the seller's revenue is not reduced.
\end{abstract}

\section{Introduction}
Marketing is a vital element in the development of the economy. Due to limited personal social connections, sellers often seek various kinds of ways to enlarge the market and attract more potential buyers. Traditionally, they tend to sell products via online shopping platforms such as Amazon and eBay. However, the platforms cannot always improve the sellers' revenue because it may cost a large amount of money for using the platforms' services such as advertisement.

Alternatively, a seller can hold an auction among her neighbours using the classic auction protocol Vickrey-Clarke-Groves (VCG)~\cite{clarke1971multipart,vickrey1961counterspeculation,groves1973incentives} with an optimal reserve price~\cite{myerson1981optimal}, which optimizes the seller's revenue locally. To further increase the seller's revenue, diffusion mechanisms on social networks have been proposed to attract more buyers. Li \textit{et al.} proposed the first such auction on social networks, called information diffusion mechanism (IDM)~\cite{DBLP:conf/aaai/LiHZZ17}. IDM can incentivize the seller's neighbours to propagate the auction information to their neighbours, and these newly informed neighbours will do the same. Eventually, all potential buyers on the network are informed, which ultimately improves the seller's revenue. To achieve this goal, IDM distributes dedicated rewards to the cut-points from the seller to the winner who receives the item. However, according to the theorem of small-world networks~\cite{amaral2000classes}, the chance for a node to be a cut-point in a well-connected network is very low. Hence, only a very small proportion of the buyers on the network can benefit from the mechanism, which disincentivizes their participation.

Therefore, in this paper, we propose another novel diffusion mechanism, which distributes the rewards to all the related buyers not only the cut-points on the paths to the winner. In addition to the cut-points, we also pay a group of buyers who are not cut-points but can disconnect the winner from the seller together with other non-cut-points. They are less important compared to the cut-points, but still critical in terms of connecting the seller and the winner. Under this mechanism, we still ensure that buyers will report their truthful valuation for the item and invite all their neighbours without a predefined reward. More importantly, we tackled the challenge without sacrificing the seller's revenue, i.e., the seller's revenue in our mechanism is at least as good as that in the previous mechanism, which incentivizes the seller to apply our mechanism.

Much work has been devoted to social networks. Granovetter first analyzed the social networks and put emphasis on the strength of weak ties, which to some extent confirms the low chance to be a cut-point~\cite{granovetter1977strength}. Matthew studied the models and techniques for analyzing
social and economic networks~\cite{jackson2010social}. Their work showed the importance of networks for social economy. Also, there is some work concerning auctions on networks. Wang and Chiu presented a recommendation system to calculate the level of recommendation for online auctions using social network analysis~\cite{wang2008recommending}. Pandit \textit{et al.} designed a system based on social networks to avoid auction fraud~\cite{pandit2007netprobe}. They mainly focused on the real-world applications without considering the mathematical properties of the mechanisms, while we look at the game theoretical properties of the mechanism on networks.

There exist some closely related work on information diffusion~\cite{easley2010networks,scott1988social}. For instance, Li \textit{et al.} gave a class of mechanisms similar to IDM and proved that IDM gives the lowest revenue in the class~\cite{DBLP:conf/ijcai/LiHZY19}. However, they still focused on distributing rewards to the cut-points while our mechanism does not belong to this class and our goal is to give rewards to all the buyers who have made a contribution to the sale. Emek \textit{et al.} studied mechanism design problem for multi-level marketing~\cite{emek2011mechanisms}. In their setting, all the nodes in a referral tree should purchase the product and they focused on false-name attacks. However, we do not have the constraint and aim for selling one item with more participants on the network to increase the seller's revenue.

In our mechanism, in order to pay the non-cut-points, we have applied some techniques from the redistribution mechanism design literature. Many redistribution mechanisms have been proposed to redistribute the surplus from the seller back to the buyers~\cite{Cavallo:2006:ODM:1160633.1160790,DBLP:journals/geb/GuoC09,DBLP:conf/aaai/Guo11}. The objective of their redistribution mechanisms is to satisfy the budget balance property, which the goal is to give back the payments to the participants as much as possible. However, in our setting, we are aiming to improve the seller's revenue through getting more potential buyers. Thus, we only borrow the idea of redistribution to reward more buyers in our mechanism while improving the seller's revenue (budget-balance would not give any revenue to the seller).

The remainder of the paper is organized as follows. In Section~\ref{section:pre}, we describe the preliminaries of the problem. In Section~\ref{section:FDM}, we define the basic concepts and introduce our novel mechanism in detail, and show the advantages of our mechanism compared to the previous work. Then we show the key properties of our mechanism in Section~\ref{section:property}. Finally, we conclude and discuss future work in Section~\ref{section:con}.

\section{Preliminaries}
\label{section:pre}
We consider a market where a seller $s$ sells an item in a social network. The network is modelled as an undirected graph $G=(V,E)$, where $V = N\cup \{s\}=\{1,2,\dots,n\}\cup\{s\}$ denotes the set of all nodes of the network and $E$ denotes the set of all the edges. Each $i\in N$ represents a potential buyer of the item, and she has a set of neighbours $r_i\subseteq V$. Node $j\in r_i$ if there is an edge $e_{ij}\in E$ connecting buyer $i$ and buyer $j$. Each buyer $i\in V$ has a depth $d_i>0$ representing the length of the shortest path from the seller to $i$.
Each buyer $i\in V$ has a private valuation $v_i\geq0$ for receiving the item. We assume that the seller's valuation for the item is zero. Figure \ref{network} shows an example of the social network, where the letter beside each node is the label of a buyer and the value in each node is the buyer's private valuation for receiving the item.

\begin{figure}[htbp]
	\centering
	\includegraphics[width=0.6\columnwidth]{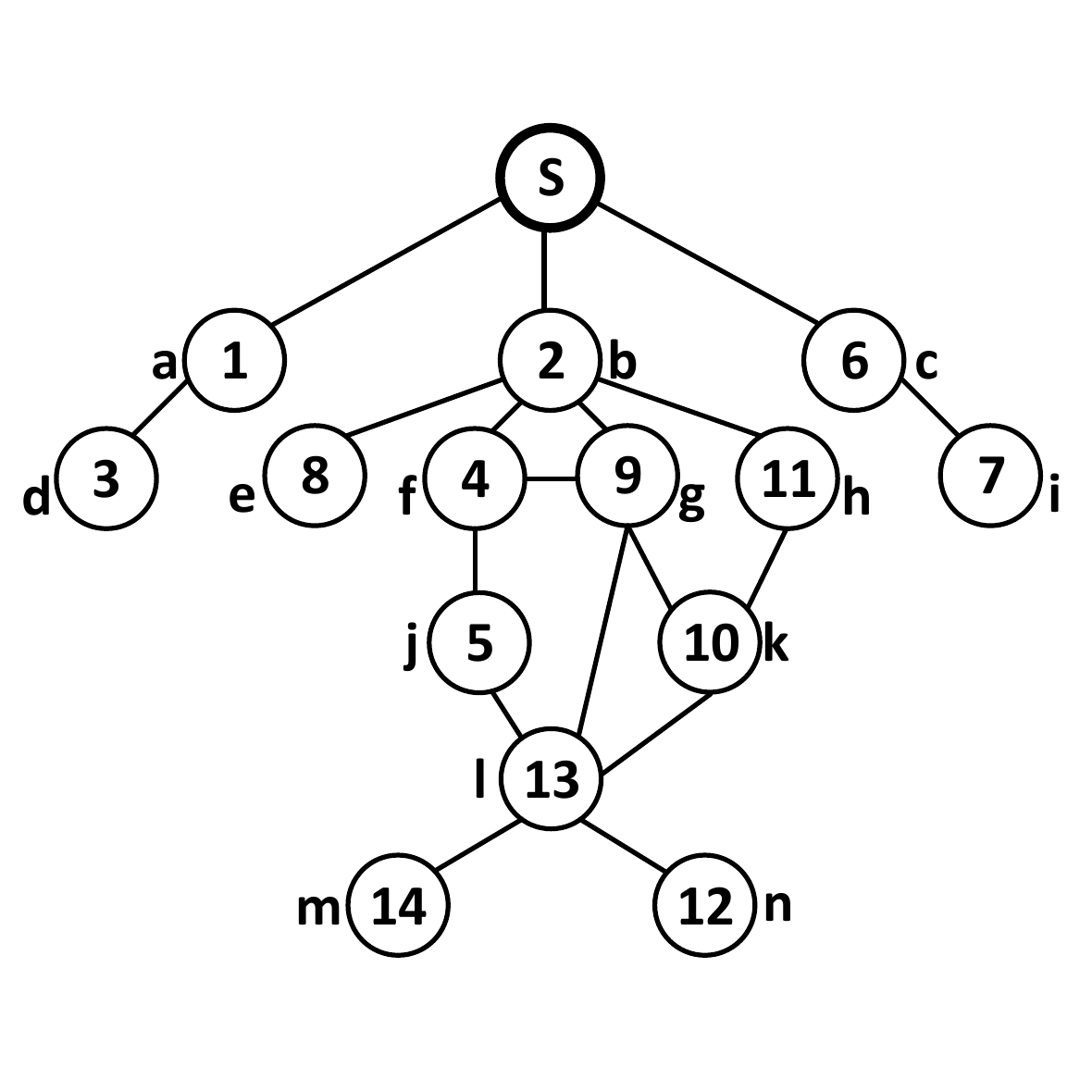}
	\caption{A social network example}
	\label{network}
\end{figure}

Traditionally, since the seller $s$ has no prior knowledge about the network, she can only sell the item among her neighbours $r_s$ without doing any advertising. In order to gain more revenue, the seller has to invite more potential buyers with higher valuations to join the sale. This can be achieved by asking her neighbours to invite their neighbours to join the sale. However, they would not invite their neighbours to compete with them without any incentive. Hence, we build incentives to tackle this challenge in this paper.

In this paper, we propose a novel diffusion mechanism, which aims to reward all buyers who make a contribution for inviting the winner. Under this mechanism, all the buyers are incentivized not only to report their private valuations for the item to the mechanism but also to propagate the sale information to all their neighbours voluntarily without prepaid rewards. 

For each buyer $i\in N$, let $\theta_i=(v_i,r_i)$ be $i$'s type and the type profile of all the buyers is denoted as $\theta=(\theta_1,\dots,\theta_n)$. Let $\theta_{-i}$ be the type profile for buyers except $i$, and we can also represent the type profile as $\theta=(\theta_{-i},\theta_i)$. Let $\Theta_i$ be the type space for buyer $i$ and $\Theta=(\Theta_1,\dots,\Theta_n)=(\Theta_{-i},\Theta_i)$ be the type profile space for all buyers.

In the mechanism, each buyer $i\in N$ is required to report her type. Let $a_i=(v_i',r_i')$ be $i$'s reported type where $v_i'$ represents the valuation she reported and $r_i'$ represents the neighbours she has invited. If she is not involved in the mechanism, let $a_i=nil$.

\begin{definition}
	We say an action profile $a$ is \textbf{feasible} if for each buyer $i\in N$ with $a_i\not=nil$, there must exist at least one path $P_{si}=(s,k_1,\dots,k_m,i)$ from $s$ to $i$, where $k_1\in r_s$, $i\in r_{k_m}'$ and $k_{t+1}\in r_{k_t}'$ for $1\leq t< m$. In other words, without invitation, the buyers cannot join to the sale. Let $\mathcal{F}(\theta)\subseteq \Theta$ be the set of all feasible action profiles.
\end{definition}

In reality, if a buyer is not invited, the mechanism will not observe any action from the buyer. To simplify the notations, instead of removing these uninformed buyers from the reports, we use "nil" action to represent their absence. Hence, feasibility always holds, and it will simplify the following definitions.

\begin{definition}
	A diffusion mechanism  $\mathcal{M}$ on the social network is defined by an allocation policy $\pi=(\pi_1, \pi_2,\dots, \pi_n)$ and a payment policy $p=(p_1,p_2,\dots, p_n)$, where $\pi_i:\Theta\rightarrow \{0,1\}$ and $p_i:\Theta\rightarrow\textbf{R}$.
\end{definition}

Given an action profile $a=(a_1,\dots,a_n)\in \mathcal{F}(\theta)$, the payment policy $p(a)=(p_1(a),\dots,p_n(a))$ represents the money each buyer is asked to pay. For buyer $i\in N$, if $p_i(a)\geq0$, then she needs to pay $p_i$ to the seller, and if $p_i(a)<0$, she will receive $|p_i(a)|$ from the seller. The allocation policy $\pi(a)=(\pi_1(a),\dots,\pi_n(a))$ represents the item allocation, and we have

$$\pi_{i}(a)=
\begin{cases}
1& \text{if buyer $i$ receives the item}\\
0& \text{if buyer $i$ does not receive the item}
\end{cases}$$

Since there is only one item to sell, we say the allocation $\pi$ is feasible if no more than one buyer with $a_i\not=nil$ receives the item under all feasible action profiles. We will only consider feasible allocation policies in the following discussion.

\begin{definition}
	Given a feasible action profile $a\in\mathcal{F}(\theta)$ and a feasible allocation $\pi$, the social welfare of allocation $\pi(a)$ is $\sum_{i\in N}\pi_{i}(a)v_i'$.
\end{definition}

Under the diffusion mechanism $\mathcal{M}=(\pi,p)$, we assume that there is no cost for a buyer to spread the sale information to her neighbours\footnote{If we consider cost for information diffusion, we will not be able to guarantee revenue improvement for the seller.}. Thus, for buyer $i\in N$ of type profile $\theta_i$, given a feasible action profile $a$ of all buyers, $i$'s utility is defined as
\begin{center}
	$u_i(\theta_i,a)=\pi_i(a)v_i-p_i(a)$
\end{center}
We say a diffusion mechanism is individually rational if the utility of each buyer involved is non-negative as long as she reports the valuation truthfully no matter how many neighbours she invites to join the mechanism. Notice that the definition does not rely on diffusion as we do not want to force people to invite others to guarantee a non-negative gain.

\begin{definition}
	A diffusion mechanism $\mathcal{M}=(\pi,p)$ is \textbf{individually rational} (IR) if $u_i(\theta_i,a)\geq0$, where $a_i=(v_i, r_i')$ for all $i\in N$, all $\theta\in\Theta$ and all $a\in \mathcal{F}(\theta)$.
\end{definition}

Traditionally, if all the buyers are willing to report their truthful valuations for the item, we say the mechanism satisfies the property of incentive compatibility or truthfulness. However, in our mechanism, buyers also need to invite their neighbours. Thus, we want to incentivize buyers not only to report their truthful valuations but also to invite all their neighbours. Therefore, we define incentive compatibility as follows.

\begin{definition}
	A diffusion mechanism $\mathcal{M}=(\pi,p)$ is \textbf{incentive compatible} (IC) if $u_i(\theta_i,(a_i,a_{-i}))\geq u_i(\theta_i,(a_i',a_{-i}'))$, for all $i\in N$, all $\theta\in\Theta$ and all $(a_i,a_{-i})\in \mathcal{F}(\theta)$ such that $a_i=\theta_i$ and for all $j\not=i$, $a_j'=a_j$ if $j$ is still connected to the seller when only $i$'s action is changed from $a_i$ to $a_i'$, otherwise $a_j'=nil$.
\end{definition}

In the following section, we will introduce a novel diffusion mechanism rewarding all the related buyers who make a contribution for inviting the winner, which satisfies the properties of IR and IC. We further prove that the revenue of the seller under our mechanism is higher than that of the traditional VCG in which the seller sells the item among her neighbours, and also higher than that of the previous work.

\section{Fair Diffusion Mechanism}
\label{section:FDM}
In this section, we will introduce our advanced mechanism called the fair diffusion mechanism (FDM). This mechanism aims to distribute rewards to all the buyers (not only the cut-points) who contributed to connect the winner with IR and IC guaranteed, which is achieved without sacrificing the seller’s revenue (even better than the previous work~\cite{DBLP:conf/aaai/LiHZZ17}).

Before we introduce our mechanism, we need some additional definitions.

\begin{definition}
	Given a feasible action profile $a\in\mathcal{F}(\theta)$, for each buyer $i\in N$, if there exists no path from the seller to $i$ without the participation of a set $D_i\subseteq N$, we say $D_i$ is a \textbf{cut set} of buyer $i$. If there is no proper subset $D_i'\subset D_i$ which is also a cut set of $i$, we say $D_i$ is a \textbf{minimal cut set} of buyer $i$.
\end{definition}

The cut sets of a buyer are the buyers who can separate the buyer from the seller and all the minimal cut sets can be induced from the cut sets. For example, in Figure~\ref{chart:node_type}, $\{b,g\}$,$\{d,f,g,h\}$ and $\{l,k\}$ are three cut sets of buyer $l$, while $\{b\}$, $\{f,g,h\}$ and $\{l\}$ are three minimal cut sets of $l$. 

\begin{definition}
	Given a feasible action profile $a\in\mathcal{F}(\theta)$, for each buyer $i,j\in N$, we say $j$ is a \textbf{critical ancestor} of $i$ if $j$ belongs to a minimal cut set of buyer $i$.
\end{definition}

\begin{definition}
	Given a feasible action profile $a\in\mathcal{F}(\theta)$, for each buyer $i,j\in N$, we say $j$ is a \textbf{strong critical ancestor} of $i$ if $j$ alone forms a minimal cut set of $i$. We say $j$ is a \textbf{weak critical ancestor} of $i$ if $j$ is a critical ancestor, but not strong critical ancestor of $i$.
\end{definition}

Intuitively, for each buyer $i\in N$, her critical ancestors are those who are on some simple path from the seller to $i$. Strong critical ancestors are the cut points from the seller to $i$, while weak critical ancestors are those who connect strong critical ancestors. All these critical ancestors have a contribution to connect the seller and buyer $i$. In Figure~\ref{chart:node_type}, all the colored nodes are the critical ancestors of buyer $m$, where nodes in orange are strong critical ancestors and nodes in yellow are weak critical ancestors.

\begin{definition}
	Given a feasible action profile $a\in\mathcal{F}(\theta)$, for each buyer $i,j\in N$, we say $j$ is $i$'s \textbf{critical descendant} if $i$ is a strong critical ancestor of $j$. Let $V_i=\{j|\text{ $j$ is i's critical descendant}, j\in N\}$ be the \textbf{critical descendant set} of buyer $i$. Similarly, for any set $K\subseteq N$, we say $j$ is $K$'s critical descendant if $K$ is a cut set of $j$. Let $V_K=\{j|\text{ $j$ is K's critical descendant}, j\in N\}$ be the critical descendant set of the set $K$.
\end{definition}
We can easily observe that on a social network, if a buyer $i$ quits the mechanism, her critical descendant set will not be involved in the mechanism. Here, we use the notation $N_{-i}$ to represent the vertex set in the new network without the participation of $i$, which is equivalent to $N\setminus V_i$. Similarly, $N_{-K}=N\setminus \bigcup_{i\in K}V_i$. Take Figure~\ref{chart:dominant_set} as an example, as all the red nodes cannot be involved in the mechanism without the participation of buyer $b$, they are critical descendant set of $b$. Thus, $V_{b}=\{b,e,f,g,h,j,k,l,m,n\}$ and $N_{-b}=\{a,c,d,i\}$.

\begin{definition}
	Given a feasible action profile $a\in\mathcal{F}(\theta)$, for each buyer $i\in N$, let $C_i$ be the \textbf{strong critical ancestor sequence} of $i$, denoted by $C_i=\{c_1^i,c_2^i,\cdots,c_k^i\}$, where $c_k^i=i$. Each $c_j^i\in C_i$ is a strong critical ancestor of buyer $i$ and the order is determined by the relation of depth $d_{c_1^i}<d_{c_2^i}<\cdots<d_{c_k^i}$.
\end{definition}

To simplify the description, let $C=\{c_1,c_2,\cdots,c_h\}$ be the strong critical ancestor sequence of the highest bidder $h$ among all the buyers on the network (with random tie-breaking). 

\begin{definition}
	Given a feasible action profile $a\in\mathcal{F}(\theta)$, for each $c_i,c_{i+1}\in C$, let $M_{c_ic_{i+1}}$ be the \textbf{weak critical ancestor set} between $c_i$ and $c_{i+1}$, denoted by $M_{c_ic_{i+1}}=\{m_{c_ic_{i+1}}^1,m_{c_ic_{i+1}}^2,\cdots,m_{c_ic_{i+1}}^k\}$, where $v_{m_{c_ic_{i+1}}^1}\geq v_{m_{c_ic_{i+1}}^2}\geq\cdots\geq v_{m_{c_ic_{i+1}}^k}$. Each $m_{c_ic_{i+1}}^j\in M_{c_ic_{i+1}}$ is a weak critical ancestor of buyer $i$, who is on some simple path from $c_i$ to $c_{i+1}$.
\end{definition}

\begin{figure}[t]
	\centering
	\subfigure[]{%
		\label{chart:node_type}%
		\includegraphics[width=0.5\columnwidth]{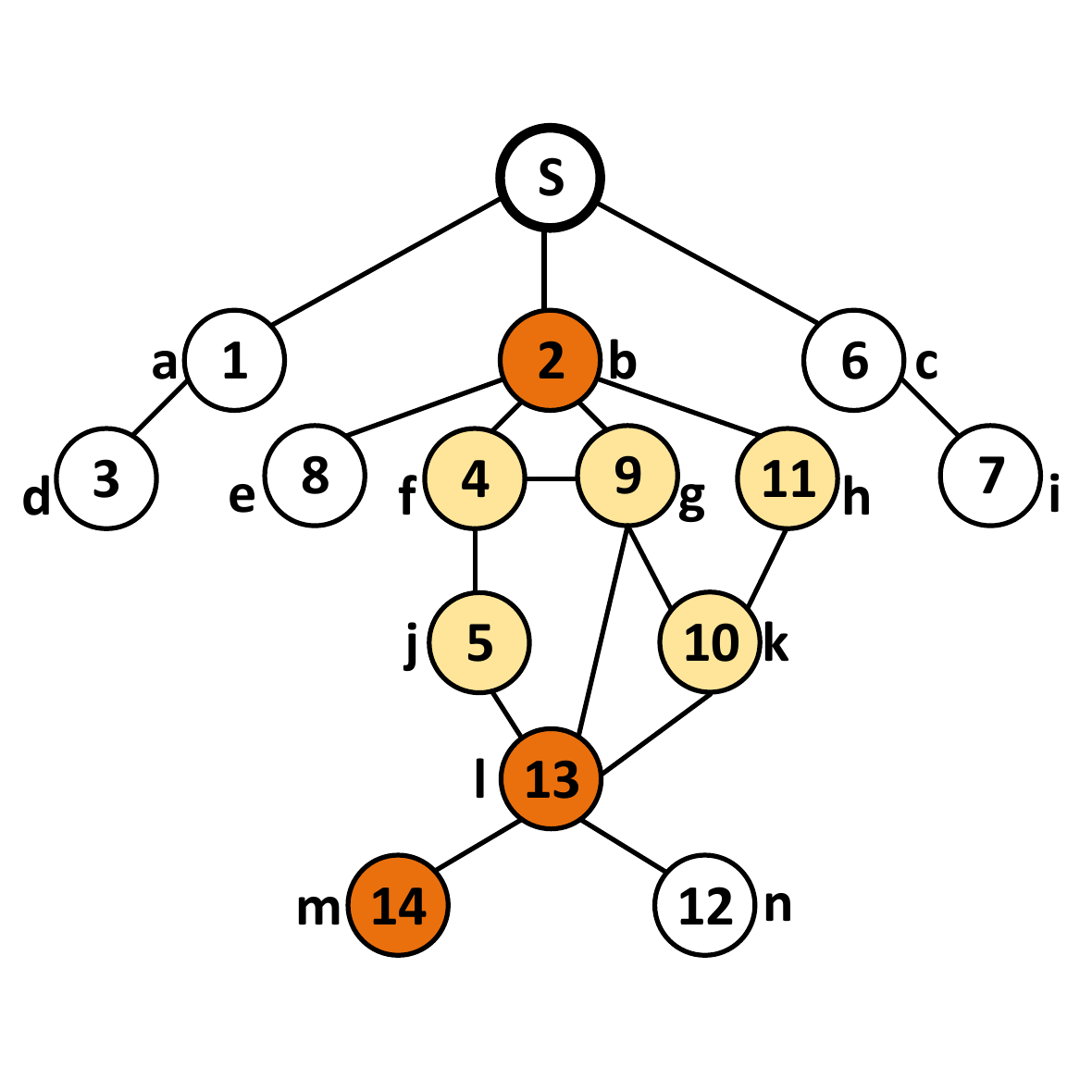}}%
	\subfigure[]{%
		\label{chart:dominant_set}%
		\includegraphics[width=0.5\columnwidth]{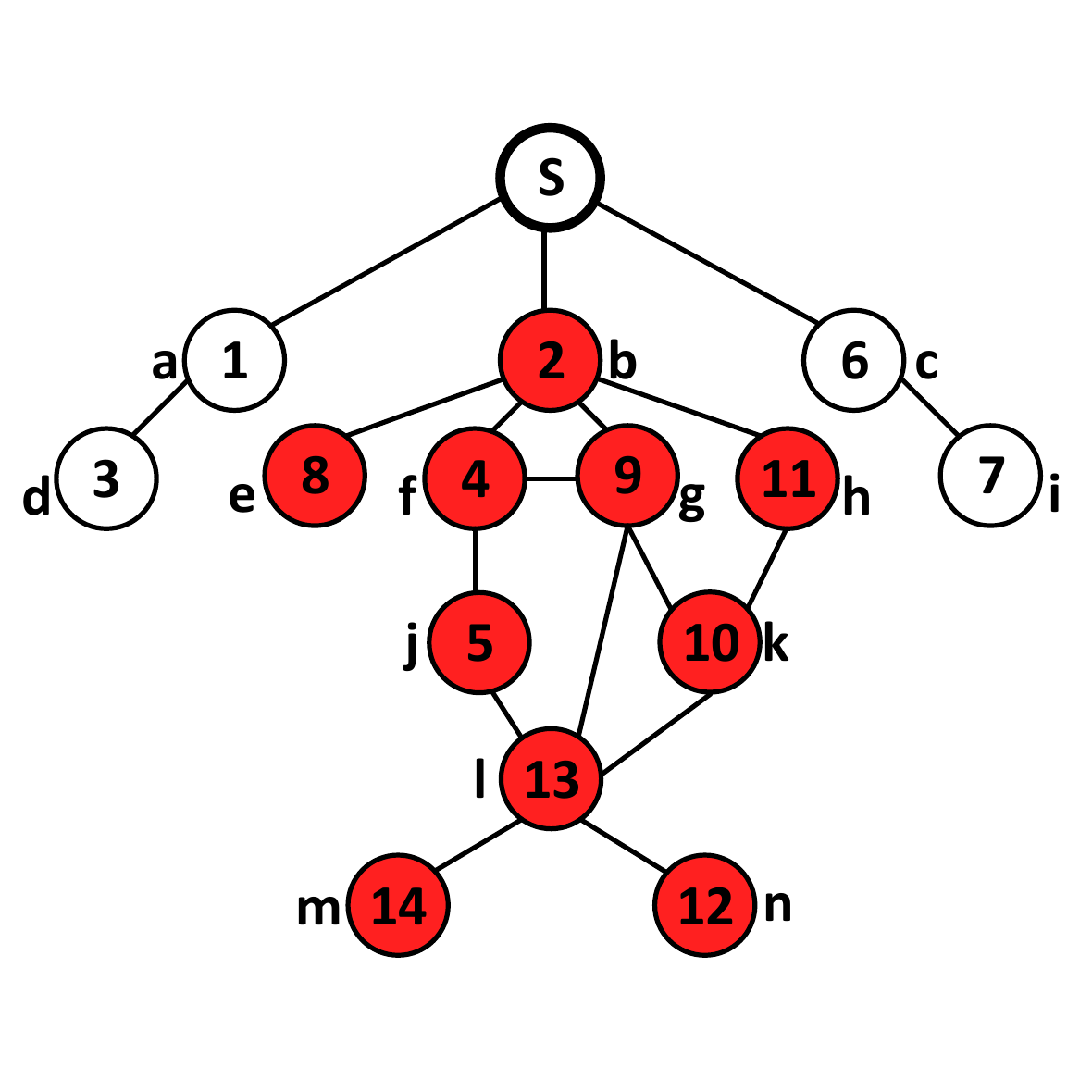}}\\
	\caption{(a) The orange nodes represent the strong critical ancestors of buyer $m$, while the yellow nodes represent the weak critical ancestors of buyer $m$ on the social network; (b) The red nodes represent the critical descendant set of buyer $b$.}
	\label{node_type}
\end{figure}

As shown in Figure~\ref{chart:node_type}, buyer $m$ is the highest bidder with reported valuation $v_m'=14$. Therefore, the strong critical ancestor sequence of $m$ is $C=\{b,l,m\}$ and the weak critical ancestor sets are $M_{bl}=\{h,k,g,j,f\}$ and $M_{lm}=\emptyset$.

Li \textit{et al.} proposed a diffusion mechanism IDM on the social network~\cite{DBLP:conf/aaai/LiHZZ17}. Their mechanism does satisfy the IC property we have defined. However, it only distributes rewards to the winner's strong critical ancestors on the network and ignores the contribution of the winner's weak critical ancestors. Therefore, only a few specific nodes may receive a non-zero utility for diffusing the information.

In contrast, the diffusion rewards in our mechanism are distributed more fairly. Especially, not only strong critical ancestors are rewarded, but also weak critical ancestors who are not cut-points but do diffuse the sale mechanism to the winner are rewarded. Moreover, the seller's revenue under our mechanism is at least as good as that in IDM. 

Our mechanism is defined as follows.

\begin{framed}
	\noindent\textbf{Fair Diffusion Mechanism (FDM)}
	
	\noindent\rule{\textwidth}{0.5pt}
	
	\begin{enumerate}
		\item Given a feasible action profile $a\in \mathcal{F}(\theta)$, find the highest bidder $h\in\arg\max_{i\in N} v_i'$ (with random tie-breaking). Let $v_D^{1^{st}}=\max_{i\in D} v_i'$ be the maximum reported valuation in the subset $D\subseteq N$, and then $v_h'=v_N^{1^{st}}$.  
		Let $g_D^{1^{st}}\in\arg\max_{i\in D}v_{V_i}^{1^{st}}$ (with random tie-breaking) be the strong critical ancestor in the subset $D\subseteq N$ of the highest bidder in $V_D$.
		
		\item Then, the allocation policy can be recursively defined as:
		
		\begin{itemize}
			\item \textbf{Allocation Policy}: 
			\begin{align*}
			\pi_{i}(a)=
			\begin{cases}		
			1& \text{if $i=c_j\in C$, $v_{i}'=v_{N_{-\{c_{j+1}\}\cup{M_{c_jc_{j+1}}}}}^{1^{st}}$}\\& \text{and $\sum_{k\in N_{-i}}\pi_{k}(a)=0$}\\
			0& \text{otherwise}
			\end{cases}
			\end{align*}
		\end{itemize}
		
		\item According to the allocation policy, we can get a winner $c_w\in C$ with $\pi_{c_w}(a)=1$. Then we distribute rewards to the buyers on the strong critical ancestor sequence $\hat{C}=\{c_1,c_2,\cdots,c_w\}$ and the weak critical ancestors $\bigcup_{j=1}^{w-1}M_{c_jc_{j+1}}$. 
		
		\item We have the payment policy defined as:
		
		\begin{itemize}	
			\item \textbf{Payment Policy}: $p_{i}=$
			$$
			\begin{cases}
			v_{N_{-{c_j}}}^{1^{st}}\!-\!v_{N_{-{\{c_{j+1}\}\cup M_{c_jc_{j+1}}}}}^{1^{st}}\!-\!R_{c_j}& \text{if $i\!=\!c_j\!\in\! \hat{C}\!\setminus\!{c_w}$}\\
			v_{N_{-{c_w}}}^{1^{st}}-R_{c_w}& \text{if $i=c_w$}\\
			-R_i& \text{if $i\in M_{c_{j-1}c_{j}}$}\\
			0& \text{otherwise}
			\end{cases}$$
			where $R_i$ is defined as: $R_{i}=$
			$$
			\begin{cases}
			\frac{v_{N_{-\{c_j\}\cup g_{M_{c_{j-1}c_j}}^{1^{st}}}}^{1^{st}}-v_{N_{-{\{c_{j}\}\cup M_{c_{j-1}c_{j}}}}}^{1^{st}}}{|M_{c_{j-1}c_j}|+1}& \text{if $i=c_j\in \hat{C}$}\\
			\frac{v_{N_{-\{i\}\cup \{c_j\}}}^{1^{st}}-v_{N_{-{\{c_{j}\}\cup M_{c_{j-1}c_{j}}}}}^{1^{st}}}{|M_{c_{j-1}c_j}|+1}& \text{if $i\in M_{c_{j-1}c_{j}}$}\\
			0& \text{otherwise}
			\end{cases}$$
		\end{itemize}
	\end{enumerate}
\end{framed}

\begin{figure}[t]
	\centering
	\subfigure[Find the strong critical ancestors.]{%
		\label{chart:FDM_process_1}%
		\includegraphics[width=0.5\linewidth]{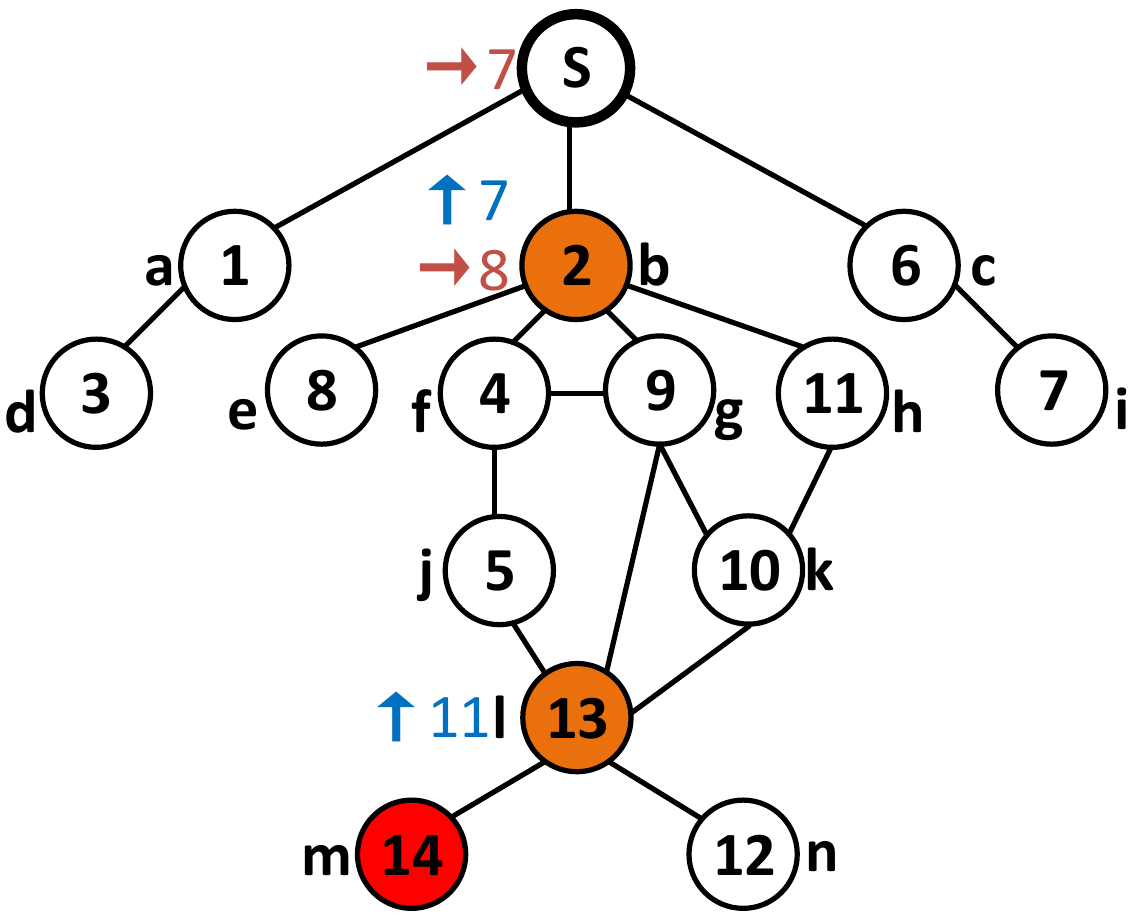}}%
	\subfigure[Redistribute the rewards.]{%
		\label{chart:FDM_process_2}%
		\includegraphics[width=0.5\linewidth]{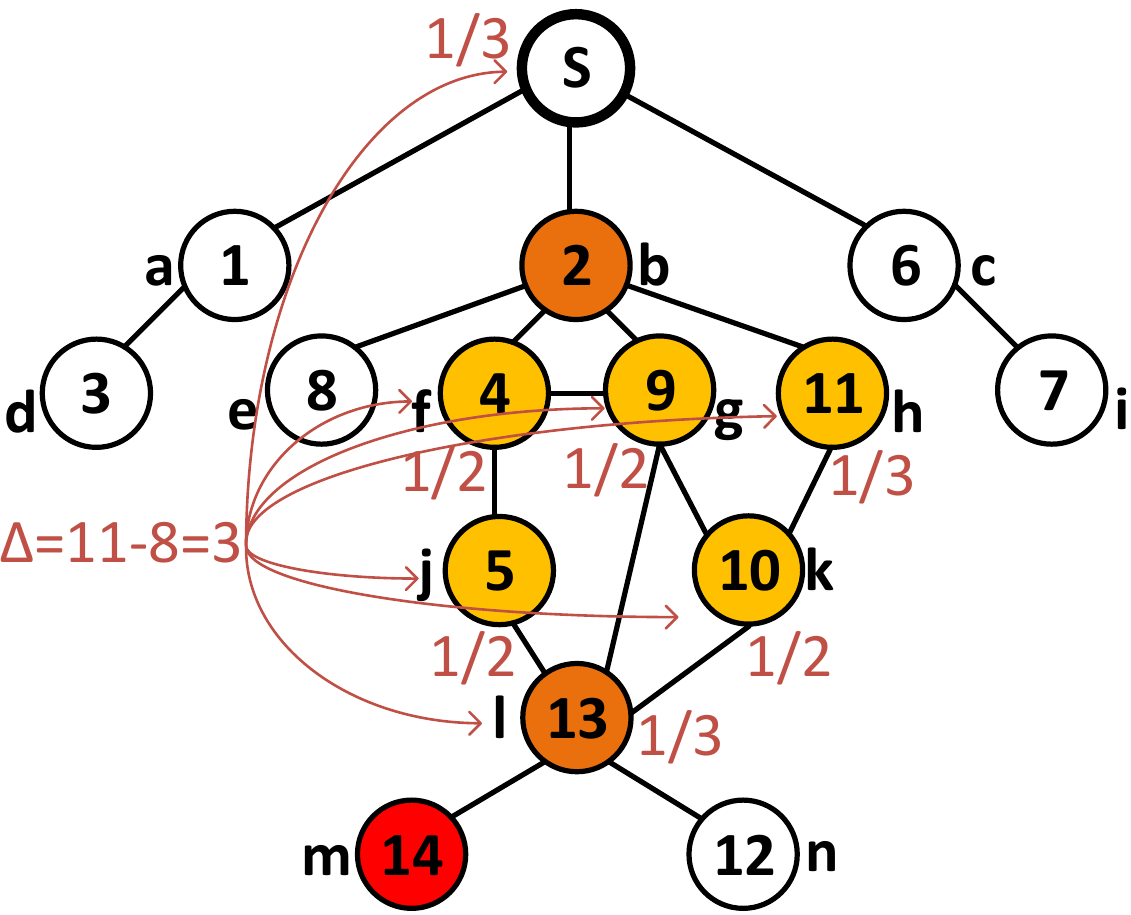}}\\
	\caption{Buyer $m$ is the highest bidder and $l$ is the winner under allocation policy. Orange nodes are the strong critical ancestors of buyer $l$ and yellow nodes are the weak critical ancestors. The difference of what $l$ pays and what $b$ receives will be redistributed among $f$, $g$, $h$, $j$, $k$ and $l$. The remaining part after redistribution will be given to the seller. }
	\label{FDM_process}
\end{figure}

The intuition behind the FDM is that the mechanism allocates the item to the first buyer $c_j$ in the strong critical ancestor sequence of the highest bidder whose reported valuation is the highest among all the buyers if $\{c_{j+1}\}\cup M_{c_jc_{j+1}}$  are not involved in the auction, where $c_{j+1}$ is the next strong critical ancestor and $M_{c_jc_{j+1}}$ is the weak critical ancestor set between $c_j$ and $c_{j+1}$. 

For each strong critical ancestor $c_j\in \hat{C}$, her payment consists of three parts: 
\begin{itemize}
	\item $v_{N_{-c_j}}^{1^{st}}$: the money she paid, which is used to distribute among herself $c_{j}$, the last strong critical ancestor $c_{j-1}$, weak critical ancestors $M_{c_{j-1}c_j}$ and the seller.
	\item $v_{N_{-{\{c_{j+1}\}\cup M_{c_jc_{j+1}}}}}^{1^{st}}$: the money she received from the next strong critical ancestor $c_{j+1}$.
	\item $R_{c_j}$: the reward she received after redistribution.
\end{itemize}

The money paid is the highest reported valuation without her participation and the money received is the highest reported valuation without the participation of $\{c_{j+1}\}\cup M_{c_jc_{j+1}}$. Specially, for the winner, the money she received is zero as there exists no next critical ancestor. 

Since the money paid and the money received between two strong critical ancestors $c_j$ and $c_{j+1}$ are not always equal, the mechanism redistributes the difference to the weak critical ancestors $M_{c_jc_{j+1}}$ and the strong critical ancestor $c_{j+1}$. Inspired by the VCG redistribution mechanism~\cite{Cavallo:2006:ODM:1160633.1160790}, the redistributed reward of buyer $i\in M_{c_jc_{j+1}}\cup \{c_{j+1}\}$ is calculated by the lower-bound of the new difference over all possible reported type of $i$ divided by the number of the buyers sharing the reward, which is $| M_{c_jc_{j+1}}|+1$ in our mechanism. In the single-item setting, the lower-bound will be achieved when $a_i'=nil$, i.e., she does not participate in the mechanism. More concretely, for $c_j\in \hat{C}$, if she quits the mechanism, the new strong critical ancestor among $M_{c_{j-1}c_j}$ is $g_{M_{c_{j-1}c_j}}^{1^{st}}$ whose payment is $v_{N_{-\{c_j\}\cup g_{M_{c_{j-1}c_j}}^{1^{st}}}}^{1^{st}}$. Then the lower-bound of the new difference is $v_{N_{-\{c_j\}\cup g_{M_{c_{j-1}c_j}}^{1^{st}}}}^{1^{st}}-v_{N_{-{\{c_{j}\}\cup M_{c_{j-1}c_{j}}}}}^{1^{st}}$. For $i\in M_{c_{j-1}c_j}$, if she quits the mechanism, the strong critical ancestor $c_j$ will remain the same but her new payment will become $v_{N_{-\{i\}\cup \{c_j\}}}^{1^{st}}$ without $i$'s participation. Then the lower-bound of the new difference is $v_{N_{-\{i\}\cup \{c_j\}}}^{1^{st}}-v_{N_{-{\{c_{j}\}\cup M_{c_{j-1}c_{j}}}}}^{1^{st}}$.

Considering the computational complexity of the mechanism, the allocation and payment for each agent can be calculated by running DFS. Thus, the total time complexity of the mechanism is $O(|V|(|V|+|E|))$, which is the same as the early mechanism.

Based on the social network discussed before, here we give a running example of FDM in Figure~\ref{FDM_process}. Among all the
buyers on the network, buyer $m$ reports the highest valuation with $v_m'=14$. Then $C=\{b,l,m\}$ is the strong critical ancestor sequence of $m$. According to the allocation policy, the item is given to buyer $l$ because $v_l'=v_{N_{- \{m\}\cup M_{lm}}}^{1^{st}}$. Thus, the strong critical ancestor sequence of $l$ is $\hat{C}=\{b,l\}$ and the weak critical ancestor set is $M_{bl}=\{h,k,g,j,f\}$. 

We first consider the strong critical ancestors. For buyer $b$, the money she pays to the seller is $v_{N_{-b}}^{1^{st}}=v_i'=7$ and the money she receives from buyer $l$ is $v_{N_{-\{l\}\cup M_{bl}}}^{1^{st}}=v_e'=8$. Similarly, for buyer $l$, the money she pays is $v_{N_{-l}}^{1^{st}}=v_h'=11$ and she receives nothing since she is the winner. Therefore, the difference between the money paid by buyer $l$ and the money received by buyer $b$ is $\Delta=11-8=3$. Then we redistribute the difference to $M_{bl}$ and $l$, and the number of buyers sharing the reward is $6$. For buyer $l$, if she does not participate in the mechanism, the winner will be $h$ whose payment is $v_k'=10$. Then the difference becomes $\Delta'=10-8=2$. So the reward to $l$ is $R_l=2/6=1/3$. Similarly, for buyer $h$, if she quits the mechanism, the strong critical ancestor is still $l$ but her payment becomes $v_k'=10$. Then the difference will also become $\Delta'=10-8=2$ without $h$'s participation. Thus, we have $R_h=2/6=1/3$. For buyer $f,g,j,k$, the difference will not change if any of them is not involved in the mechanism, so we have $R_f=R_g=R_j=R_k=3/6=1/2$. Till now, $M_{bl}$ and $l$ are redistributed the rewards. The remaining money $\Delta-\sum_{i\in M_{bl} \cup \{l\}} R_i=3-2*1/3-4*1/2=1/3$ will be given to the seller. Then the payment for all the critical buyers is calculated as: $p_b=7-8=-1$, $p_l=11-1/3=32/3$, $p_h=-1/3$ and $p_f=p_g=p_j=p_k=-1/2$. According to the definition of utility, we have $u_b=\pi_{b}(a)v_b-p_b=0-(-1)=1$, $u_l=13-32/3=7/3$, $u_h=0-(-1/3)=1/3$ and $u_f=u_g=u_j=u_k=0-1/2=1/2$. The revenue of the seller is $u_s = p_b+p_l+p_h+p_f+p_g+p_j+p_k= (-1)+32/3+(-1/3)+4*(-1/2) = 22/3$.

\subsection{Comparison between FDM and IDM}

\begin{figure}[t]
	\centering
	\subfigure[A running example of FDM.]{%
		\label{chart:FDM_result}%
		\includegraphics[width=0.5\linewidth]{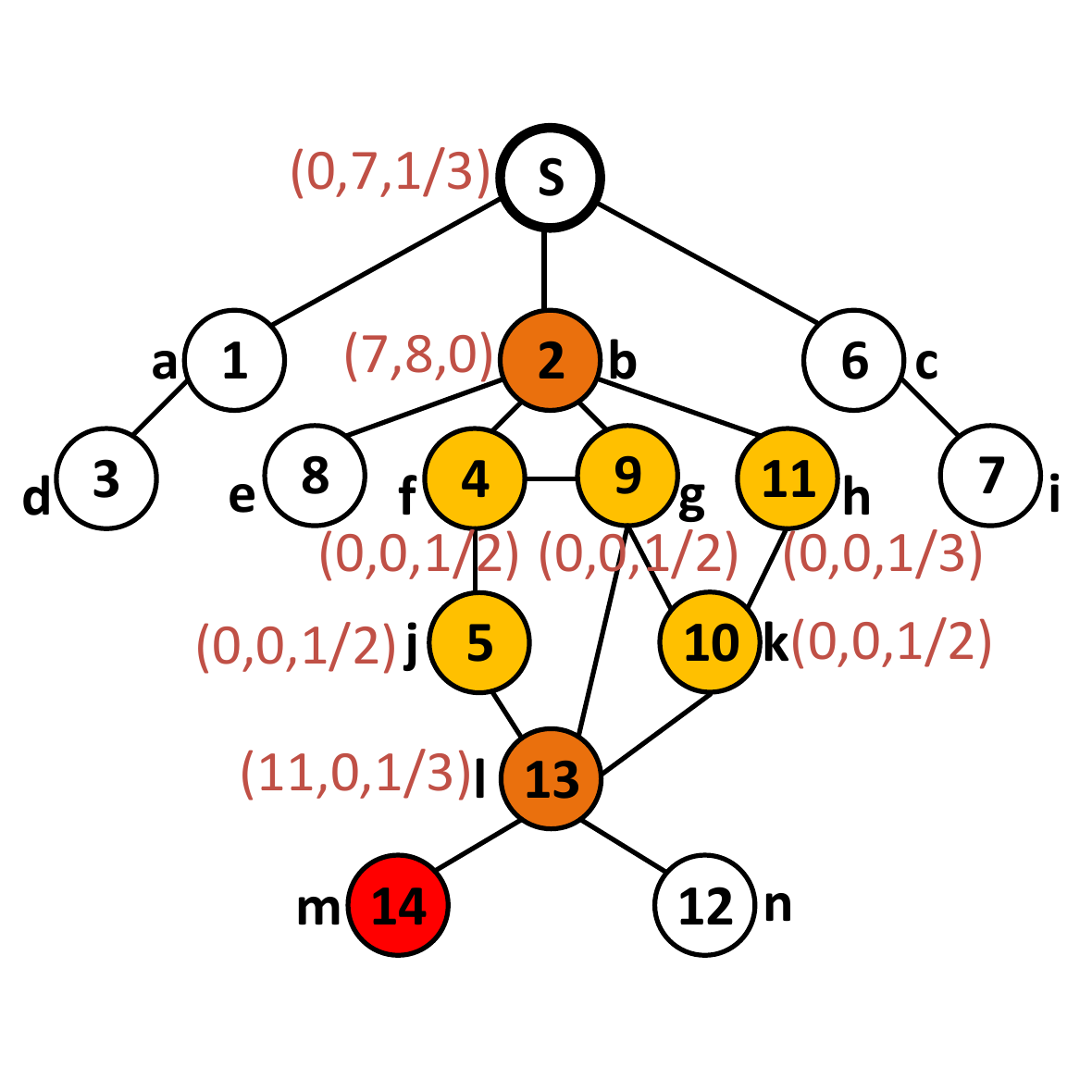}}%
	\subfigure[A running example of IDM.]{%
		\label{chart:IDM_result}%
		\includegraphics[width=0.5\linewidth]{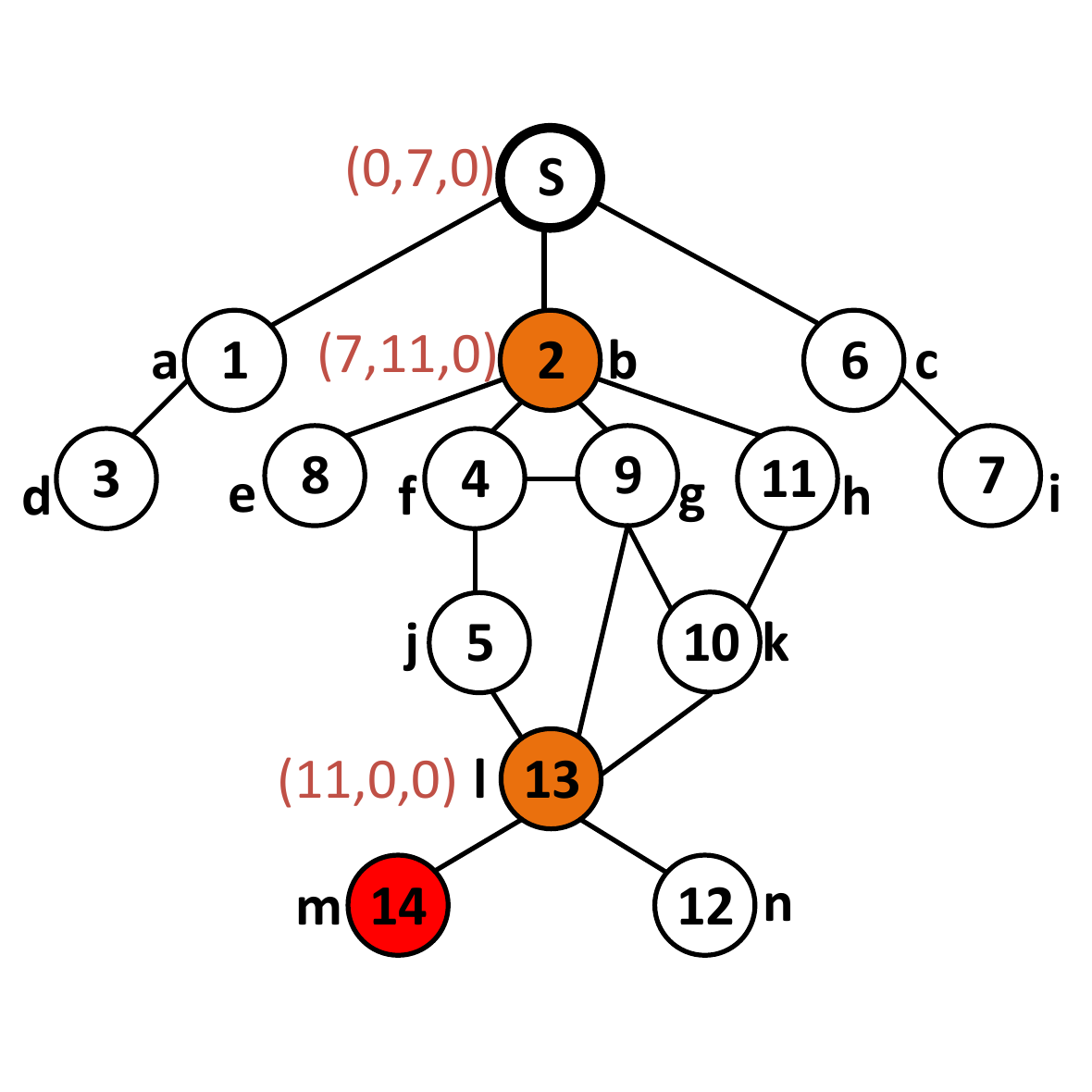}}\\
	\caption{7 critical ancestors in FDM have positive utilities while only 2 strong critical ancestors have positive utilities in IDM. For each buyer, the first number in the vector is the money paid, the second is the money received and the third is the money redistributed.}
	\label{Comparison}
\end{figure}

To show the advantages of our mechanism over the previous related work, here we compare the running result for the same example under IDM and analyze the performance of FDM and IDM.

Compared to the IDM proposed by Li \textit{et al.}~\cite{DBLP:conf/aaai/LiHZZ17}, they only give rewards to those strong critical ancestors. As the running example shown in Figure~\ref{chart:IDM_result} for the same social network, under IDM, buyer $l$ is also the winner with $p_l = 11$ and the utility of strong critical ancestor $b$ is $4$, but all the other buyers on the network will have zero utilities. This is because none of the other buyers is a cut-point to reach $l$ from $s$. In contrast, five more buyers $f,g,h,j$ and $k$, who are also on some simple path from the seller to the winner, gain positive utilities in the same setting under FDM in Figure~\ref{chart:FDM_result}. Although they are not cut-points to reach $l$ from $s$, they can disconnect $l$ from $s$ together. Therefore, FDM also considers their diffusion contribution from this aspect, which is fairer for all the buyers who have made contributions to the sale in the network. Moreover, the seller's revenue under FDM is $22/3$ which is greater than $7$ in IDM.

\begin{table}
	\centering
	\begin{tabular}{c|c|c}
		\toprule[2pt]
		&\textbf{FDM}  &\textbf{IDM}\\ \midrule[1pt]
		winner &        buyer $l$  & buyer $l$\\ 
		social welfare &        13  &  13\\ 
		beneficial buyers  &        $b,f,g,h,j,k,l$  &$b,l$\\ 
		\# of beneficial buyers  &        \textbf{7}  &2\\ 
		beneficial critical ancestor ratio  &        \textbf{1}  &0.29\\
		buyers' total utility  &        5.67  &\textbf{6}\\
		seller's revenue  &        \textbf{7.33}  &7\\
		\bottomrule[2pt]
	\end{tabular}
	\caption{The performance difference of FDM and IDM.}
	\label{tab:comparison}
\end{table}

Table~\ref{tab:comparison} gives an intuitive display for the performance difference between FDM and IDM of the same running example. We can obviously observe that although the winner and the social welfare are the same for the two mechanisms, the number of beneficial buyers in FDM is far greater than that of IDM. Furthermore, the beneficial critical ancestor ratio, i.e., the percentage of positive-utility buyers over all the critical ancestors, is 1 in FDM while 0.29 in IDM. This indicates that all the critical ancestors of the winner in FDM but only a tiny fraction in IDM will be rewarded, which shows the fairness of our mechanism. Moreover, under any network, fixing a maximal value node, if the others' valuation follows some distribution, then all critical ancestors' expected utility in our mechanism is greater than zero, which cannot be achieved by IDM. On the other hand, FDM does not sacrifice the seller's revenue. In spite of the decrease of the buyers' total utility, the seller's revenue under FDM is much higher than IDM with all the desirable properties guaranteed, which encourages the seller more to apply our mechanism.

\section{Properties of FDM}
\label{section:property}
In this section, we will prove that FDM satisfies the properties of IR and IC, and the seller's revenue is at least as good as the revenue under IDM, which is no less than that of traditional VCG among neighbours.

Firstly, we show that all the buyers in our mechanism will not have negative utilities if they report their valuation truthfully.
\begin{thm}
	The fair diffusion mechanism is individually rational.
\end{thm}

\begin{proof}
	After the execution of the FDM, only critical buyers may have non-zero utilities. Since $R_i$ is the redistributed reward, it is obvious that $R_i\geq0$ according to the definition. 
	\begin{itemize}
		\item For buyer $i=c_j\in \hat{C}\setminus c_w$, her utility is $u_{c_j}(a)=\pi_{c_j}(a)v_{c_j}-p_{c_j}=v_{N_{-{\{c_{j+1}\}\cup M_{c_jc_{j+1}}}}}^{1^{st}}-v_{N_{-{c_j}}}^{1^{st}}+R_{c_j}$. Since buyer $c_j$ is ahead of $M_{c_jc_{j+1}}$ on any path from seller to the winner, we have $V_{c_j}\supset V_{\{c_{j+1}\}\cup M_{c_jc_{j+1}}}$ and then $N_{-\{c_{j+1}\}\cup M_{c_jc_{j+1}}}\supset N_{-c_j}$. Thus, we have $u_{c_j}(a)=v_{N_{-{\{c_{j+1}\}\cup M_{c_jc_{j+1}}}}}^{1^{st}}-v_{N_{-{c_j}}}^{1^{st}}+R_{c_j}\geq v_{N_{-{\{c_{j+1}\}\cup M_{c_jc_{j+1}}}}}^{1^{st}}-v_{N_{-{c_j}}}^{1^{st}}\geq0$.
		\item For buyer $i=c_w$, since she is the winner, according to the allocation policy, we have $v_{c_w}=v_{N_{-{\{c_{w+1}\}\cup M_{c_wc_{w+1}}}}}^{1^{st}}$. Then her utility is $u_{c_w}(a)=\pi_{c_w}(a)v_{c_w}-p_{c_w}=v_{c_w}-v_{N_{-{c_w}}}^{1^{st}}+R_{c_w}\geq v_{c_w}-v_{N_{-{c_w}}}^{1^{st}}=v_{N_{-{\{c_{w+1}\}\cup M_{c_wc_{w+1}}}}}^{1^{st}}-v_{N_{-{c_w}}}^{1^{st}}\geq0$.
		\item For buyer $i\in M_{c_{j-1}c_{j}}$, her utility is $u_{i}(a)=\pi_{i}(a)v_{i}-p_{i}=R_i\geq0$.
	\end{itemize}
	
	The payments for all the other buyers are zero. Therefore, the FDM is individually rational. 
\end{proof}

Theorem~\ref{IC} proves that in FDM, all the buyers are incentivized to report their truthful type to the seller, i.e., their truthful valuations and all their neighbours.

\begin{thm}\label{IC}
	The fair diffusion mechanism is incentive compatible.
\end{thm}
\begin{proof}
	According to the definition of incentive compatibility, we have to prove that for all the buyers in the graph, reporting their truthful valuations for the item and propagating the sale information to all their neighbours is the dominant strategy. Note that we do not consider the collaboration between buyers. More concretely, each buyer $i \in N$ can only cut the edges to the neighbours who also belong to her critical descendant set, which is $r_i\cap V_i$, because those neighbours can only receive the sale information from $i$. However, the other neighbours can receive the information from the seller through other paths to connect $i$, then buyer $i$ cannot cut the edges by herself.
	
	The buyers on the network can be divided into four groups in FDM: 
	
	\begin{enumerate}[(1)]
		\item the non-winner strong critical ancestors $c_j\in \hat{C}\setminus c_w$.
		\item the weak critical ancestors between strong critical ancestors.
		\item the winner $c_w$ who receives the item. 
		\item all the other buyers who are not in group (1), (2) and (3).
	\end{enumerate}
	
	\begin{itemize}
		\item \textbf{For any strong critical ancestor $c_j$ in Group (1)}: 
		\begin{itemize}
			\item If the neighbour set $r'_{c_j}$ reported is fixed, the utility of buyer $c_j$ is defined by $u_{c_j}=\frac{v_{N_{-\{c_j\}\cup g_{M_{c_{j-1}c_j}}^{1^{st}}}}^{1^{st}}-v_{N_{-{\{c_{j}\}\cup M_{c_{j-1}c_{j}}}}}^{1^{st}}}{|M_{c_{j-1}c_j}|+1}+v_{N_{-{\{c_{j+1}\}\cup M_{c_jc_{j+1}}}}}^{1^{st}}-v_{N_{-{c_j}}}^{1^{st}}$, which is not related to her reported valuation $v_{c_j}'$. According to the allocation policy, she cannot misreport valuation to become a weak critical ancestor. If she is still the strong critical ancestor, then the allocation is unchanged. Thus no matter what valuation she reports, her utility remains the same. If she reports a higher valuation to be the winner, the reward redistributed remains the same but her utility will decrease: $u_{c_j}'=v_{c_j}-v_{N_{-{c_j}}}^{1^{st}}+R_{c_j}<v_{N_{-{\{c_{j+1}\}\cup M_{c_jc_{j+1}}}}}^{1^{st}}-v_{N_{-{c_j}}}^{1^{st}}+R_{c_j}=u_{c_j}$. 
			\item If the valuation $v'_{c_j}$ reported is fixed and $r_{c_j}'\not=r_{c_j}$. 
			\begin{itemize}
				\item If she is still the strong critical ancestor, we have $N'_{-\{c_{j+1}\}\cup M_{c_jc_{j+1}}'}\subseteq N_{- \{c_{j+1}\}\cup M_{c_jc_{j+1}}}$, and then $v_{N'_{-{\{c_{j+1}\}\cup M'_{c_jc_{j+1}}}}}^{1^{st}}\leq v_{N_{-{\{c_{j+1}\}\cup M_{c_jc_{j+1}}}}}^{1^{st}}$. Since the money paid $v_{N_{-c_j}}^{1^{st}}$ and the reward redistributed $R_{c_j}$ remains the same, removing neighbours may decrease the money she received.
				\item If she becomes a weak critical ancestor with positive utility, then her utility becomes $u'_{c_j}\leq\frac{v_{N_{-\{c_j\}\cup g_{M_{c_{j-1}c_j}}^{1^{st}}}}^{1^{st}}-v_{N_{-{\{c_{j}\}\cup M_{c_{j-1}c_{j}}}}}^{1^{st}}}{|M_{c_{j-1}c_j}|+1}\leq u_{c_j}$, since $v_{N_{-{\{c_{j+1}\}\cup M_{c_jc_{j+1}}}}}^{1^{st}}-v_{N_{-{c_j}}}^{1^{st}}\geq0$.
				\item If she becomes the new winner, the reward redistributed remains the same and her utility becomes $u_{c_j}'=v_{c_j}-v_{N_{-{c_j}}}^{1^{st}}+R_{c_j}<v_{N_{-{\{c_{j+1}\}\cup M_{c_jc_{j+1}}}}}^{1^{st}}-v_{N_{-{c_j}}}^{1^{st}}+R_{c_j}=u_{c_j}$.
				\item If she is neither a strong critical ancestor nor a weak critical ancestor, her utility $u_{c_j}'=0$.
			\end{itemize}
		\end{itemize} 
		
		\item \textbf{For any weak critical ancestor $i\in M_{c_{j-1}c_j}$ in Group (2)}: 
		\begin{itemize}
			\item If the neighbour set $r'_{i}$ reported is fixed, the utility of buyer $i$ is defined by $u_{i}=R_i=\frac{v_{N_{-\{i\}\cup c_j}}^{1^{st}}-v_{N_{-{\{c_{j}\}\cup M_{c_{j-1}c_{j}}}}}^{1^{st}}}{|M_{c_{j-1}c_j}|+1}$, which is not related to her reported valuation $v_{i}'$. According to the allocation policy, she cannot misreport valuation to become a strong critical ancestor. If the allocation is unchanged, no matter what valuation she reports, her utility remains the same. If she reports a higher valuation to be the winner, the reward redistributed to her remains the same and her utility becomes $u_{i}'=v_{i}-v_{N_{-{i}}}^{1^{st}}+\frac{v_{N_{-\{i\}\cup c_j}}^{1^{st}}-v_{N_{-{\{c_{j}\}\cup M_{c_{j-1}c_{j}}}}}^{1^{st}}}{|M_{c_{j-1}c_j}|+1}$. Since $v_i<v_{N_{-{i}}}^{1^{st}}$, we have $u_i'<u_i$.
			\item If the valuation $v'_{i}$ reported is fixed and $r_{i}'\not=r_{i}$. The utility of buyer $i$ is $u_i=R_i$, which is not related to $i$'s critical descendants. Since removing some neighbours cannot change the allocation and cannot decrease the number of buyers sharing the reward with buyer $i$, misreporting neighbours will not increase the utility.
		\end{itemize} 
		
		\item \textbf{For the winner $c_w$ in Group (3)}: 
		\begin{itemize}
			\item If the neighbour set $r'_{c_w}$ reported is fixed, the utility of buyer $c_w$ is defined by $u_{c_w}=v_{c_w}+\frac{v_{N_{-\{c_w\}\cup g_{M_{c_{w-1}c_w}}^{1^{st}}}}^{1^{st}}-v_{N_{-{\{c_{w}\}\cup M_{c_{w-1}c_{w}}}}}^{1^{st}}}{|M_{c_{w-1}c_w}|+1}-v_{N_{-{c_w}}}^{1^{st}}$, which is not related to her reported valuation $v_{c_w}'$. If the allocation is unchanged, no matter what valuation she reports, her utility remains the same. If she reports a lower valuation to be a weak critical ancestor, her utility becomes $u_{c_w}'\leq\frac{v_{N_{-\{c_w\}\cup g_{M_{c_{w-1}c_w}}^{1^{st}}}}^{1^{st}}-v_{N_{-{\{c_{w}\}\cup M_{c_{w-1}c_{w}}}}}^{1^{st}}}{|M_{c_{w-1}c_w}|+1}\leq u_{c_w}$, since $v_{c_w}-v_{N_{-{c_w}}}^{1^{st}}\geq0$. If she becomes a strong critical ancestor of her critical descendants, her utility will be $u'_{c_w}=v_{N_{-\{c_{w+1}\}\cup M_{c_wc_{w+1}}}}^{1^{st}}+\frac{v_{N_{-\{c_w\}\cup g_{M_{c_{w-1}c_w}}^{1^{st}}}}^{1^{st}}-v_{N_{-{\{c_{w}\}\cup M_{c_{w-1}c_{w}}}}}^{1^{st}}}{|M_{c_{w-1}c_w}|+1}-v_{N_{-{c_w}}}^{1^{st}}<u_{c_w}$ since $v_{c_w}\geq v_{N_{-\{c_{w+1}\}\cup M_{c_wc_{w+1}}}}^{1^{st}}$.
			\item If the valuation $v'_{c_w}$ reported is fixed and $r_{c_w}'\not=r_{c_w}$, the utility of buyer $c_w$ is defined by $u_{c_w}=v_{c_w}+\frac{v_{N_{-\{c_w\}\cup g_{M_{c_{w-1}c_w}}^{1^{st}}}}^{1^{st}}-v_{N_{-{\{c_{w}\}\cup M_{c_{w-1}c_{w}}}}}^{1^{st}}}{|M_{c_{w-1}c_w}|+1}-v_{N_{-{c_w}}}^{1^{st}}$, which is not related to her neighbours $r_{c_w}'$. According to the allocation policy, $v'_{c_w}=v_{N_{-\{c_{w+1}\}\cup{M_{c_wc_{w+1}}}}}^{1^{st}}$. Thus the allocation will not be changed, and no matter what neighbourhood she reports, her utility remains the same.
		\end{itemize} 
		
		\item \textbf{For any other buyer $i$ in Group (4)}: 
		\begin{itemize}
			\item If the neighbour set $r'_{i}$ reported is fixed, the utility of buyer $i$ is zero. If $c_w$ is not her strong critical ancestor, the only way she can gain some benefits is to report a higher valuation to win the item. However, if she reports $v_i'> v_N^{1^{st}} > v_i$, her payment will be the original maximum valuation on the network, which is greater than her truthful valuation. If $c_w$ is her strong critical ancestor, no matter what valuation she reports, the allocation will not be changed.
			\item If the valuation $v'_i$ reported is fixed and $r_{i}'\not=r_{i}$, removing some neighbours will not change the allocation.
		\end{itemize} 
	\end{itemize}
	
	In summary, we can draw the conclusion that the FDM is incentive compatible.
\end{proof}

Then we prove that although our FDM can distribute rewards to more related buyers who have made contributions to the sale, it will not sacrifice the seller's revenue. Actually, it can even improve the seller's revenue compared to the previous work IDM with all the properties guaranteed, which encourages the seller to apply our mechanism.

\begin{thm}
	The seller's revenue under fair diffusion mechanism is always at least as good as the revenue under IDM, which is no less than that of traditional VCG among neighbours.
\end{thm}

\begin{proof}
	Given a feasible action profile $a \in \mathcal{F}(\theta)$, the seller's revenue is the sum of the first critical ancestor $c_1$'s payment and all the rewards redistributed to the seller. $R_s$ is the remaining money of the difference which is not redistributed among the critical ancestors. It is easy to confirm that $R_s\geq0$.
	\begin{align*}
		u_s^{FDM}(a,(\pi,p))=&\sum_{i\in N}p_i(a,(\pi,p))\\=&v_{N_{-c_1}}^{1^{st}}+R_s\\\geq& v_{N_{-c_1}}^{1^{st}}
	\end{align*}
	
	While under the IDM, the seller's revenue is defined by $u_s^{IDM}(a,(\pi,p))=v_{N_{-c_1}}^{1^{st}}$. Thus, we have
	\begin{align*} 
		u_s^{FDM}(a,(\pi,p))\geq& v_{N_{-c_1}}^{1^{st}}=u_s^{IDM}(a,(\pi,p))\\\geq& v_{r_s}^{2^{nd}}=u_s^{VCG}(a,(\pi,p))
	\end{align*}

	Therefore, the seller's revenue in FDM is non-negative and at least as good as that in IDM, which is also no less than that in traditional VCG among neighbours.
	
\end{proof}

Since the seller's revenue is the sum of the first strong critical ancestor's payment and the reward redistributed to her, we can easily observe that the seller's revenue in IDM is the lower bound of that in FDM.

\section{Conclusion}
\label{section:con}
In this paper, we propose an advanced diffusion mechanism on social networks. The seller can run the mechanism without paid third-party platforms and gain a higher revenue. Our mechanism guarantees that participating buyers are incentivized to offer their truthful valuations for the item and invite all their neighbours to the sale. All the related critical buyers on some simple path from the seller to the winner will be rewarded for their diffusion effort, which is fairer than other mechanisms proposed in previous work. Moreover, the seller's revenue in our mechanism will not be sacrificed, and is even improved compared to other related work.

On the basis of our work, many other problems are worth further investigation. One direction is to generalize FDM to a more complex setting for multiple items~\cite{Zhao:2018:SMI:3237383.3237400}. Since the item can be passed through the critical ancestors in FDM, it gives us a good chance to study a distributed method to realize our mechanism. What's more, the false-name attack is a difficult problem in mechanism design. False-name attacks also exist in our network setting. We find it also worthwhile to consider the Bayesian Nash equilibrium to maximize the seller's revenue if given a valuation distribution on social networks~\cite{jain2008efficient,hartline2008optimal}. Another valuable future work can be generalizing our mechanism to broader settings such as weighted networks to achieve the same goal~\cite{DBLP:conf/ijcai/LiHZY19}. FDM considers to reward all buyers on the simple paths to reach the winner. What about the others who are not on these paths, but their valuations play an important role to determine the payments? Should they also be rewarded?

\bibliography{ecai}
\end{document}